\documentclass{article}
\usepackage{amsmath}
\usepackage{amssymb}
\usepackage{amsthm}
\usepackage{color}
\usepackage{graphicx}
\usepackage{subfigure}
\usepackage{gensymb}
\usepackage{multicol}
\usepackage[affil-it]{authblk}
\usepackage[margin=1in]{geometry}
\usepackage{thmtools}
\usepackage{thm-restate}

\usepackage{verbatim} 

\newtheorem{claim}{Claim}

\title{\Large \textbf{Bayesian Network Structure Learning Using Quantum Annealing}}

\author[1,2]{\normalsize Bryan A. O'Gorman}
\author[1]{Alejandro Perdomo-Ortiz}
\author[2]{Ryan Babbush}
\author[2]{\newline Al\'{a}n Aspuru-Guzik}
\author[1]{Vadim Smelyanskiy}
\affil[1]{\normalsize Quantum Artificial Intelligence Laboratory,
NASA Ames Research Center, Moffett Field, CA}
\affil[2]{Department of Chemistry and Chemical Biology, 
Harvard University, Cambridge, MA}

\date{\today}

\begin{document}
\maketitle

\begin{abstract}
We introduce a method for  the problem of learning the structure of a Bayesian network using the quantum adiabatic algorithm. We do so by introducing an efficient reformulation of a standard posterior-probability scoring function on graphs as a pseudo-Boolean function, which is equivalent to a system of 2-body Ising spins, as well as suitable penalty terms for enforcing the constraints necessary for the reformulation; our proposed method requires $\mathcal O(n^2)$ qubits for $n$ Bayesian network variables. Furthermore, we prove lower bounds on the necessary weighting of these penalty terms. The logical structure resulting from the mapping has the appealing property that it is instance-independent for a given number of Bayesian network variables, as well as being independent of the number of data cases.
\end{abstract}

\begin{multicols}{2}

\section{Introduction}
\label{sec:introduction}
Bayesian networks are a widely used probabilistic graphical model in machine learning \cite{koller2009}. A Bayesian network's structure encapsulates conditional independence within a set of random variables, and, equivalently, enables a concise, factored representation of their joint probability distribution. There are two broad classes of computational problems associated with Bayesian networks: inference problems, in which the goal is to calculate a probability distribution or the mode thereof given a Bayesian network and the state of some subset of the variables; and learning problems, in which the goal is to find the Bayesian network most likely to have produced a given set of data. Here, we focus on the latter, specifically the problem of Bayesian network structure learning. Bayesian network structure learning has been applied in fields as diverse as the short-term prediction of solar-flares \cite{yu2010} and the discovery of gene regulatory networks \cite{djebbari2008, friedman2000}. The problem of learning the most likely structure to have produced a given data set, with reasonable formal assumptions to be enumerated later, is known to be \textsc{NP}-complete \cite{chickering1996}, so its solution in practice requires the use of heuristics. 

Quantum annealing is one such heuristic. Though efficient quantum algorithms for certain problems are exponentially faster than their classical counterpart, it is believed that quantum computers cannot efficiently solve \textsc{NP}-complete problems \cite{aaronson2010}. However, there exist quantum algorithms (for problems not known or believed to be \textsc{NP}-complete) that have a provable quadratic speedup over classical ones \cite{grover1996, somma2012}. There is therefore reason to believe quantum-mechanical effects such as tunneling could provide a polynomial speedup over classical computation for some sets of problems. The recent availability of quantum annealing devices from D-Wave Systems has sparked interest in the experimental determination of whether or not the current generation of the device provides such speedup \cite{ronnow2014, venturelli2014}. While there exists prior work related to ``quantum Bayesian networks''\cite{tucci2012} and the ``quantum computerization'' of classical Bayesian network methods \cite{tucci2014}, the results presented here are unrelated.

In this paper, we describe how to efficiently map a certain formulation of \textsc{Bayesian Network Structure Learning} (\textsc{BNSL}) to \textsc{Quadratic Unconstrained Binary Optimization} (\textsc{QUBO}). The \textsc{QUBO} formalism is useful because it is mathematically equivalent to that of a set Ising spins with arbitrary 2-body interactions, which can be mapped to the Ising spins with a limited 2-body interaction graph as implementable by physical quantum annealing devices. Similar mappings have been developed and implemented for lattice protein folding \cite{babbush2014b, perdomo2012}, planning and scheduling \cite{rieffel2014}, fault diagnosis \cite{perdomo2014}, graph isomorphism \cite{gaitan2014}, training a binary classifier \cite{babbush2014a, denchev2013}, and the computation of Ramsey numbers \cite{bian2013}.

To map \textsc{BNSL} to \textsc{QUBO}, we first encode all digraphs using a set of Boolean variables, each of which indicates the presence or absence of an arc (i.e. directed edge), and define a pseudo-Boolean function on those variables that yields the score of the digraph encoded therein so long as it satisfies the necessary constraints. This function is not necessarily quadratic, and so we apply standard methods to quadratize (i.e. reduce the degree to two) using ancillary variables. We then introduce ancillary variables and add additional terms to the pseudo-Boolean function corresponding to constraints, each of which is zero when the corresponding constraint is satisfied and positive when it is not. The resulting \textsc{QUBO} instance is defined over $\mathcal O(n^2)$ Boolean variables when mapped from a \textsc{BNSL} instance with $n$ Bayesian network variables. Interestingly, the structure of the \textsc{QUBO} is instance-independent for a fixed \textsc{BNSL} size. Because embedding the structure of \textsc{QUBO} into physical hardware is generally computationally difficult, this is an especially appealing feature of the mapping.

We also show sufficient lower bounds on penalty weights used to scale the terms in the Hamiltonian that penalize invalid states, like those containing a directed cycle or with parent sets larger than allowed. In a physical device, setting the penalty weights too high is counterproductive because there is a fixed maximum energy scale. The stronger the penalty weights, the more the logical energy spectrum is compressed, which is problematic for two reasons: first, the minimum gap, with which the running time of the algorithm scales inversely, is proportionally compressed, and, second, the inherently limited precision of a physical device's implementation of the interaction strengths prevents sufficient resolution of logical states close in energy as the spectrum is compressed. 

The utility of the mapping from \textsc{BNSL} to \textsc{QUBO} introduced here is not limited to quantum annealing. Indeed, the methods used here were motivated by a previous mapping of the same problem to weighted \textsc{MAX-SAT} \cite{cussens2008}. Existing simulated annealing code is highly optimized \cite{isakov2014} and may be applied to \textsc{QUBO} instances derived from our mapping. In that case, there is no need to quadratize, because simulated annealing does not have the limitation to 2-body interactions that physical devices do. With respect to penalty weights, while simulated annealing does not have the same gap and precision issues present in quantum annealing, there may still be reason to avoid setting the penalty weights too high. Because the bits corresponding to arcs with different i.e. terminal vertices do not interact directly, many valid states are separated by invalid ones, and so penalty weights that are too strong may erect barriers that tend to produce basins of local optima. While simulated annealing directly on digraph structures is possible, mapping to \textsc{QUBO} and performing simulated annealing in that form has the advantage that it enables the exploitation of existing, highly optimized code, as well as providing an alternative topology of the solution space and energy landscape.

\textsc{BNSL} has a special property that makes it especially well-suited for the application of heuristics such as QA: Unlike in other problems where
anything but the global minimum is undesirable or those in which an
approximate solution is sufficient, in \textsc{BNSL} there is utility in
having a set of high scoring DAGs.
The scoring function encodes the posterior
probability, and so sub- but near-optimal solution may be almost as probable
as the global optimum.
In practice, because quantum annealing is an inherently stochastic procedure,
it is run many times for the same instance, producing a set of low-energy states. In cases where the BN structure is learned for the purpose of doing inference on it, a high-scoring subset of many quantum annealing runs can utilized by performing Bayesian model averaging, in which inference is done on
the set of likely BNs and the results averaged proportionally.

In Section \ref{sec:background}, we review the formalism of Bayesian networks and \textsc{BNSL} (\ref{sec:bnsl}) and quantum annealing (\ref{sec:qa}), elucidating the features that make the latter suitable for finding solutions of the former. 
In Section \ref{sec:mapping}, we develop an efficient and instance-independent mapping from \textsc{BNSL} to \textsc{QUBO}. In Section \ref{sec:penalty-weights}, we provide sufficient lower bounds on the penalty weights in the aforementioned mapping. In Section \ref{sec:conclusion}, we discuss useful features of the mapping and conclude. In the Appendix, we prove the sufficiency of the lower bounds given; the methods used to do so may be useful in mappings for other problems.

\section{Background}
\label{sec:background}

\subsection{\textsc{Bayesian Network Structure Learning}}
\label{sec:bnsl}

A Bayesian network (BN) is a probabilistic graphical model for a set of random variables that encodes their joint probability distribution in a more compact way and with fewer parameters than would be required otherwise by taking into account conditional independences among the variables. It consists of both a directed acyclic graph (DAG) whose vertices correspond to the random variables and an associated set of conditional probabilities for each vertex.  Here and throughout the literature, the same notation is used for both a random variable and its corresponding vertex, and the referent will be clear from context.

Formally, a BN $B$ for $n$ random variables ${\mathbf X = (X_i)_{i=1}^n}$ is a pair $(B_S, B_P)$, where $B_S$ is a DAG representing the structure of the network and $B_P$ is the set of conditional probabilities $\{p(X_i | \Pi_i(B_S))|1\leq i \leq n\}$ that give the probability distribution for the state of a variable $X_i$ conditioned on the joint state of its parent set $\Pi_i(B_S)$ (those variables for which there are arcs in the structure $B_S$ from the corresponding vertices to that corresponding to $X_i$; we will write simply $\Pi_i$ where the structure is clear from context). Let $r_i$ denote the number of states of the variable $X_i$ and $q_i=\prod_{j\in \Pi_i}r_j$ denote the number of joint states of the parent set $\Pi_i$ of $X_i$ (in $B_S$). Lowercase variables indicate realizations of the corresponding random variable; $x_{ik}$ indicates the $k$-th state of variable $X_i$ and $\pi_{ij}$ indicates the $j$-th joint state of the parent set $\Pi_i$.
The set of conditional probabilities $B_P$ consists of $n$ probability distributions $\left((\mathbf \theta_{ij})_{j=1}^{q_i}\right)_{i=1}^n$, where $\mathbf \theta_{ij} =  (\theta_{ijk})_{k=1}^{r_i}$ is the conditional probability distribution for the states $(x_{ik})_{k=1}^{r_i}$  of the variable $X_i$ given the joint state $\pi_{ij}$ of its parents $\Pi_i$ (i.e. $p(x_{ik}|\pi_{ij}) = \theta_{ijk}$).

Given a database $D = \{\mathbf x_i | 1\leq i \leq N\}$ consisting of $N$ cases, where each $\mathbf x_i$ denotes the state of all variables $\mathbf X$, the goal is to find the structure that maximizes the posterior distribution $p(B_S|D)$ out of all possible structures. By Bayes's Theorem,
\begin{equation}
p(B_S|D) = \frac{p(D|B_S)p(B_S)}{p(D)}.
\end{equation}
The marginal probability of the database $p(D)$ is the same for all structures, so assuming that each structure is equally likely, this simplifies to
\begin{equation}
p(B_S|D) \propto p(D|B_S).
\end{equation}
In Section \ref{sec:prior-info}, we describe how to account for certain types of non-uniform prior distributions over the graph structures.
With certain further reasonable assumptions, namely multinomial sampling, parameter independence and modularity, and Dirichlet priors, the latter conditional probability is
\begin{equation}
\label{eq:likelihood}
p(D|B_S) = \prod_{i=1}^n \prod_{j=1}^{q_i} \frac{\Gamma(\alpha_{ij})}{\Gamma(N_{ij} + \alpha_{ij})} \prod_{k=1}^{r_i} \frac{\Gamma(N_{ijk} + \alpha_{ijk})}{\Gamma(\alpha_{ijk})},
\end{equation}
where $N_{ijk}$ is the number of cases in $D$ such that variable $X_i$ is in its $k$-th state and its parent set $\Pi_i$ is in its $j$-th state, $N_{ij}=\sum_{k=1}^{r_i}N_{ijk}$, $\alpha_{ijk}$ is the hyperparameter for $\theta_{ijk}$ in the Dirichlet distribution from which $\mathbf{\theta}_{ij}$ is assumed to be drawn, and $\alpha_{ij} = \sum_{k=1}^{r_i}\alpha_{ijk}$ \cite{heckerman1995}.

Given a database $D$, our goal is equivalent to that of finding the structure with the largest likelihood, i.e. the structure that yields the largest probability of the given database conditioned on that structure. We do this by encoding all structures into a set of bits and defining a quadratic pseudo-Boolean function on those bits and additional ancillary bits whose minimizing bitstring encodes the structure with the largest posterior probability.

\subsection{Quantum Annealing}
\label{sec:qa}
Quantum annealing is a method for finding the minimum value of a given objective function. It is the quantum analogue of classical simulated annealing, where the computation is driven by quantum, rather than thermal, fluctuations \cite{kadowaki1998}. A similar procedure is called adiabatic quantum computation, in which the adiabatic interpolation of a Hamiltonian whose ground state is easily prepared to one whose ground state encodes the solution to the desired optimization problem guarantees that final state is indeed the ground state of the latter \cite{farhi2000}. The formalism for both is similar, and the methods described here are useful for both. Specifically, the time-dependent Hamiltonian is 
\begin{equation}
H(t) = A(t) H_0 + B(t) H_1,
\end{equation}
for $0\leq t \leq T$,
where $H_0$ is the initial Hamiltonian, $H_1$ is the final Hamiltonian, $A(t)$ is a real monotonic function such that $A(0)=1$ and $A(T)=0$, and $B(t)$ is a real monotonic function such that $B(0)=0$ and $B(T)=1$. The adiabatic theorem states that if the system starts in the ground state of $H_0$ and $H(t)$ varies slowly enough, then the system will be in the ground state of $H_1$ at time $T$. Using this procedure to solve an optimization problem entails the construction of $H_1$ such that its ground state encodes the optimal solution. In practice, arbitrary Hamiltonians are difficult to implement, but this is ameliorated by  results showing the ability to effectively implement arbitrary Hamiltonians using physically-realizable connectivity through various gadgetry with reasonable overhead \cite{oliveira2008, kaminsky2004}. 

The main contribution of this paper is a construction of $H_1$ such that its ground state encodes the solution for a given instance of \textsc{BNSL}. Specifically, we construct an instance of \textsc{QUBO} whose solution is the score-maximizing DAG; there is a simple transformation between a classically defined \textsc{QUBO} instance and a diagonal quantum 2-local Hamiltonian consisting of only Pauli $Z$ and $ZZ$ terms \cite{perdomo2008}.

When the desired Hamiltonian is diagonal and 2-local an embedding technique called graph-minor embedding can be used \cite{choi2008, choi2011}. A graph $G$ is a minor of another graph $H$ if there exists a mapping from vertices of $G$ to disjoint, individually connected subgraphs of $H$ such that for every edge $e$ in $G$ there is an edge in $H$ whose adjacent vertices are mapped to by the adjacent vertices of the edge $e$. The desired Hamiltonian and hardware are considered as graphs, called the logical and physical respectively, where qubits correspond to vertices and edges correspond to a 2-body interaction, desired or available. Graph-minor embedding consists of two parts: finding a mapping of the logical vertices to sets of physical as described, and setting the parameters of the physical Hamiltonian such that the logical fields are distributed among the appropriate physical qubits and there is a strong ferromagnetic coupling between physical qubits mapped to my the same logical qubit so that they act as one. Determining the graph-minor mapping, or even if the logical graph is a minor of the physical one, is itself \textsc{NP}-hard, and so in practice heuristics are used \cite{cai2014}.

\section{Mapping BNSL to QUBO}
\label{sec:mapping}
We use $n(n-1)$ bits $\mathbf d = (d_{ij})_{\substack{1\leq i < j \leq n\\i\neq j}}$
to encode each of the possible arcs in a directed graph, where $d_{ij}=1$ indicates the presence of the arc from vertex $X_i$ to vertex $X_j$ and $d_{ij}=0$ indicates its absence. In this way, the matrix whose entries are $\{d_{ij}\}$ is the adjacency matrix of a directed graph (where $d_{ii}=0$). Let $G(\mathbf d)$ be that directed graph encoded in some $\mathbf d$. The mapping consists of the construction of a function of these ``arc bits'' that is equal to the logarithm of the score of the structure they encode, as well as a function that penalizes states that encode graphs with directed cycles. Additionally, due to resource constraints, we add a function that penalizes structures in which any node has more than $m$ parents and allow that the scoring function only works on states that encode structures in which each vertex has at most $m$ parents.

\subsection{Score Hamiltonian}
For numerical efficiency, it is the logarithm of the likelihood for a given structure that is actually computed in practice. The likelihood given in Equation~\ref{eq:likelihood} decomposes into a product of likelihoods for each variable, which we exploit here. Let
\begin{align}
\lefteqn{s_i(\Pi_i(B_S))} \nonumber \\
& \equiv -\log \left( 
\prod_{j=1}^{q_i}\frac{\Gamma(\alpha_{ij})}{\Gamma(N_{ij} + \alpha_{ij})} 
\prod_{k=1}^{r_i} \frac{\Gamma(\alpha_{ijk} + N_{ijk})}{\Gamma(\alpha_{ijk})}
\right),
\end{align}
i.e. the negation of the ``local'' score function, and
\begin{equation}
s(B_S) \equiv \sum_{i=1}^n s_i(\Pi_i(B_S)),
\end{equation}
so that
\begin{equation}
\log p(D|B_S) = -s(B_S) = -\sum_{i=1}^n s_i(\Pi_i(B_S)).
\end{equation}
The negation is included because while we wish to maximize the likelihood, in \textsc{QUBO} the objective function is minimized.
We wish to define a quadratic pseudo-Boolean function $H_{\text{score}}(\mathbf d)$ such that $H_{\text{score}}(\mathbf d) = s(G(\mathbf d))$. Let $\mathbf d_i\equiv (d_{ji})_{\substack{1\leq j \leq n\\j \neq i}}$ and define
\begin{equation}
H_{\text{score}}(\mathbf d)\equiv \sum_{i=1}^n H_{\text{score}}(\mathbf d_i).  
\end{equation}
Any pseudo-Boolean such as $H_{\text{score}}^{(i)}$ has a unique multinomial form and $s_i(\Pi_i(G(\mathbf d)))$ depends only on arcs whose head is $X_i$ (i.e. those encoded in $\mathbf d_i$), so we write without loss of generality
\begin{equation}
\label{eq:genHscore}
H_{\text{score}}^{(i)} (\mathbf d_i)
= \sum_{J \subset \{1,\cdots, n\} \setminus \{i\}} \left( w_i(J) \prod_{j \in J} d_{ji}\right).
\end{equation}
From this it is clear that $w_i(\emptyset) = s_i(\emptyset)$. If $X_i$ has a single parent $X_j$, then the above simplifies to
\begin{equation}
H_{\text{score}}^{(i)} = w_i(\emptyset) + w_i(\{j\}) = s_i(\{X_j\}),
\end{equation}
which yields $w_i(\{j\}) = s_i(\{X_j\}) - s_i(\emptyset)$ for arbitrary $j$. Similarly, if $X_i$ has two parents $X_j$ and $X_k$, then
\begin{equation} \begin{array}{rl}
H_{\text{score}}^{(i)} &= w_i(\emptyset) + w_i(\{j\}) + w_i(\{k\}) + w_i(\{j, k\}) \\
&= s_i(\emptyset) + (s_i(\{X_j\}) - s_i(\emptyset)) \\
&\quad+ (s_i(\{X_j\}) - s_i(\emptyset)) + w_i(\{j, k\}) \\
&= s_i(\{X_j\}) + s_i(\{X_k\}) - s_i(\emptyset) + w_i(\{j, k\}) \\
&= s_i(\{X_j, X_k\}), 
\end{array} \end{equation}
which yields $w_i(\{j, k\}) = s_i(\{X_j, X_k\}) - s_i(\{X_j\}) - s_i(\{X_k\}) + s_i(\emptyset)$. Extrapolating this pattern, we find that 
\begin{equation}
w_i(J) = \sum_{l=0}^{|J|} (-1)^{|J|-l}\sum_{\substack{K\subset J \\ |K| = l}} s_i(K).
\end{equation}
Note that the general form given in Equation~\ref{eq:genHscore} includes terms of order $(n-1)$. Ultimately, we require a quadratic function and reducing high-order terms to quadratic requires many extra variables. Therefore, we limit the number of parents that each variable has to $m$ via $H_{\text{max}}$, described below, and allow that the score Hamiltonian actually gives the score only for structures with maximum in-degree $m$:
\begin{equation}
H_{\text{score}}^{(i)} (\mathbf d_i)
= \sum_{\substack{J \subset \{1,\cdots, n\} \setminus \{i\} \\ |J|\leq m}} \left( w_i(J) \prod_{j \in J} d_{ji}\right),
\end{equation}
which is equal to $s_i(\Pi_i(G(\mathbf d)))$ if $|\mathbf d_i|\leq m$.

\subsection{Max Hamiltonian}
Now we define a function $H_{\text{max}}^{(i)}$ whose value is zero if variable $X_i$ has at most $m$ parents and positive otherwise. This is done via a slack variable $y_i$ for each node. Define
\begin{align}
\label{eq:H_max-def}
d_i & \equiv |\mathbf d_i| = \sum_{\substack{1\leq j \leq n \\ j \neq i}} d_{ji},
\intertext{i.e. $d_i$ is the in-degree of $x_i$,}
\mu & \equiv \left\lceil \log_2(m+1)\right\rceil,\\
\intertext{i.e. $\mu$ is the number of bits needed to represent an integer in $[0,m]$,}
y_i & \equiv \sum_{l=1}^{\mu}2^{l-1}y_{il},\\
\intertext{i.e. $y_i\in\mathbb Z$ is encoded using the $\mu$ bits ${\mathbf y_i = (y_{il})_{l=1}^{\mu} \in \mathbb B^{\mu}}$, and}
H_{\text{max}}^{(i)}(\mathbf d_i, \mathbf y_i) & = \delta_{\text{max}}^{(i)}(m - d_i - y_i)^2,
\end{align}
where $\delta_{\text{max}}^{(i)}>0$ is the weight of the penalty. For convenience, we also write $H_{\text{max}}^{(i)}(d_i, y_i)$ without loss of generality. When viewed as a quadratic polynomial of $y_i$, $H_{\text{max}}^{(i)}$ takes its minimal value of zero when $y_i=m-d_i$. Note that $0\leq y_i\leq 2^{\mu} -1$. 
If $d_i \leq m$, let $y_i^*$ be such that $0\leq y_i^*=m-d_i \leq m \leq 2^{\mu}-1$. Then $H_{\text{max}}(d_i, y_i^*) = 0$. However, when $d_i > m$, because $y_i \geq 0$, we cannot set $y_i$ in that way. 
By taking the derivative with respect to $y_i$,
\begin{equation}
\frac{\partial }{\partial y_i}H_{\text{max}}^{(i)}(d_i, y_i)
=2\delta_{\text{max}}^{(i)}(y_i - m + d_i)>0,
\end{equation}
we see that $H_{\text{max}}^{(i)}$ takes its minimum value over the domain of $y_i$ when $y_i=0$, and that value is
\begin{equation}
H_{\text{max}}^{(i)}(d_i, 0) = \delta_{\text{max}}^{(i)}(m-d_i)^2.
\end{equation}
Noting that $H_{\text{max}}^{(i)}$ is nonnegative,
\begin{equation}
\min_{y_i}H_{\text{max}}^{(i)}(d_i, y_i) = \begin{cases}
0, & d_i \leq m,\\
\delta_{\text{max}}(d_i - m)^2, & d_i > m.
\end{cases}
\end{equation}
Thus, if the constraint $|\mathbf d_i|\leq m$ is satisfied, $H_{\text{max}}^{(i)}$ does nothing, but if $|\mathbf d_i| > m$, a penalty of at least $\delta_{\text{max}}^{(i)}$ is added.

\subsection{Acyclicity}
Lastly, we must ensure that the structure encoded in $\{d_{ij}\}$ has no directed cycles. We do so by introducing additional Boolean variables $\mathbf r=(r_{ij})_{1\leq i < j \leq n}$ that will encode a binary relation on the set of variables. Every directed acyclic graph admits at least one topological order of the vertices, and a graph with a directed cycle admits none. A topological order ``$\leq$'' of the vertices $\{X_i\}$ of a digraph is a total order thereon such that for every edge $(i, j)$ in the digraph $X_i \leq X_j$. Such an order is not unique in general. Let $r_{ij}=1$ represent $x_i \leq x_j$ and $r_{ij}=0$ represent $x_i \geq x_j$.

To ensure acyclicity, we define a function $H_{\text{trans}}(\mathbf r)$ such that $H_{\text{trans}}(\mathbf r)$ is zero if the relation encoded in $\{r_{ij}\}$ is transitive and is positive otherwise, as well as a function $H_{\text{consist}}$ such that $H_{\text{consist}}(\mathbf d$ is zero if the order encoded in $\{r_{ij}\}$ is consistent with the directed edge structure encoded by $\{d_{ij}\}$ and positive otherwise.
First, we ensure that $\{r_{ij}\}$ is transitive. Because if a tournament has any cycle, it has a cycle of length three, it is sufficient to penalize directed 3-cycles. Define
\begin{equation}
H_{\text{trans}}(\mathbf r) \equiv \sum_{1\leq i < j \leq n}H_{\text{trans}}^{(ijk)}(r_{ij}, r_{ik}, r_{jk}),
\end{equation}
where
\begin{align}
\lefteqn{H_{\text{trans}}^{(ijk)}(r_{ij}, r_{ik}, r_{jk}) } \nonumber\\
&\equiv \delta_{\text{trans}}^{(ijk)}\left[r_{ij}r_{jk}(1-r_{ik}) 
+ (1-r_{ij})(1-r_{jk})r_{ik}\right] \nonumber\\
&= \delta_{\text{trans}}^{(ijk)}\left(r_{ik}+r_{ij}r_{jk} - r_{ij}r_{ik} - r_{jk}r_{ik} \right)  \nonumber\\
&=\begin{cases} \delta_{\text{trans}}^{(ijk)}, &
\left[(x_i \leq x_j \leq x_k \leq x_i) \right.\\
& \quad \left.\lor (x_i \geq x_j \geq x_k \geq x_i)\right],\\
0, & \text{otherwise}, \end{cases}
\end{align}
and $\delta_{\text{trans}}^{(ijk)}$ is the penalty weight added if $\mathbf r$ encodes either 3-cycle containing $\{x_i, x_j, x_k\}$. Note that the superscripted indices on the penalty weight variable are unordered so that $\delta_{\text{trans}}^{(i'j'k')}\equiv \delta_{\text{trans}}^{(ijk)}$ for all permutations $(i',j',k')$ of $(i,j,k)$. 

Second, we must penalize any state that represents an order and a directed graph that are inconsistent with each other, i.e. in which $r_{ij}=1$ and $(x_j, x_i)\in E(G(\mathbf d))$ or $r_{ij}=0$ and ($(x_i, x_j) \in E(G(\mathbf d))$. Equivalently, we want to ensure that neither $r_{ij}=d_{ji}=1$ nor $r_{ij}=1-d_{ij}=0$. Define
\begin{equation}
H_{\text{consist}}(\mathbf d, \mathbf r)\equiv 
\sum_{1\leq i < j \leq n} H_{\text{consist}}^{(ij)}(d_{ij}, d_{ji}, r_{ij})
\end{equation}
and
\begin{align}
\lefteqn{H_{\text{consist}}(d_{ij}, d_{ji}, r_{ij}) } \nonumber\\
&= \delta_{\text{consist}}^{(ij)}(d_{ji}r_{ij} + d_{ij}(1-r_{ij})) \nonumber \\
&= \begin{cases} 
\delta_{\text{consist}}^{(ij)}, &
d_{ji}=r_{ij}=1 \lor (d_{ij}=1 \land r_{ij}=0),\\
0, & \text{otherwise},
\end{cases}\end{align}
which has the desired features. Again the superscripted indices on the penalty weight variable are unordered, so that $\delta_{\text{consist}}^{(ji)}\equiv \delta_{\text{consist}}^{(ij)}$ for $1\leq i < j \leq n$.
Finally, define
\begin{equation}
H_{\text{cycle}}(\mathbf d, \mathbf r)\equiv 
H_{\text{consist}}(\mathbf d, \mathbf r) + H_{\text{trans}}(\mathbf r),
\end{equation}
which takes on its minimal value of zero if $G(\mathbf d)$ is a DAG and is strictly positive otherwise.

\subsection{Total Hamiltonian}
\begin{figure*}
	\includegraphics[width=\textwidth,trim=1in 1in 1in 3in]{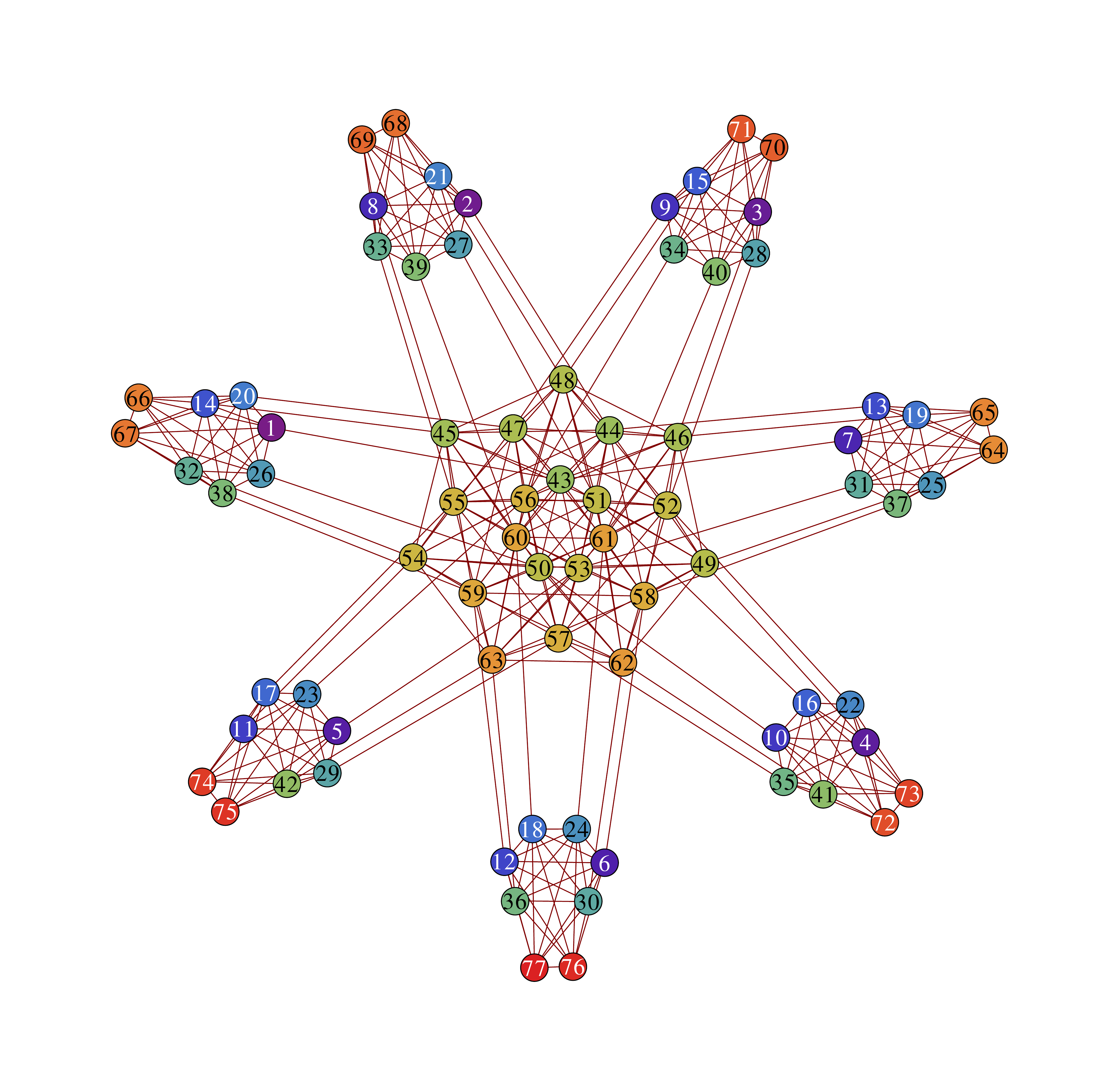}
	\includegraphics[width=\textwidth]{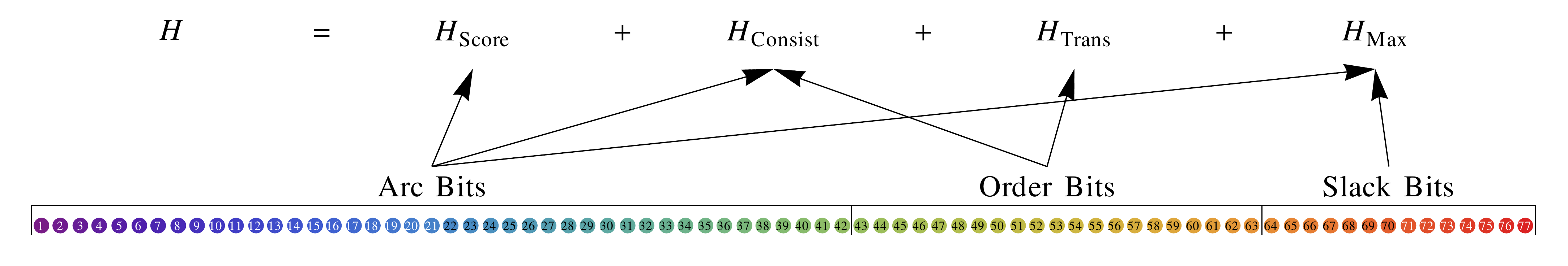}
	\caption{
Top: Logical Graph for $n=7$ BN variables with a maximum of $m=2$ parents. Each vertex corresponds to a bit in the original QUBO and an edge between two vertices indicates a non-zero quadratic term containing the corresponding bits. The central cluster is the order bits used to enforce acyclicity; it is highly connected but not complete. Each "spike" corresponds to a variable $X_i$ in the Bayesian network. The outer two vertices are the corresponding slack bits $\{y_{il}\}$ and the remaining inner vertices are the arc bits $\{d_{ji}\}$ representing those arcs for which the corresponding Bayesian network variable is the head. Each spike is a clique, due to $H_{\text{max}}$ (independent of which the arc bits for a given BN variable are fully connected due to $H_{\text{score}}$). Each arc bit is connected to a single order bit and each order bit is connected to two arc bits, due to $H_{\text{consist}}$. 
Bottom: Schematic of the Hamiltonian. The row of disks represents all of the bits in the original QUBO problem, colored consistently with the logical graph above. They are grouped into three sets: the arc bits representing the presence of the possible arcs, the order bits representing a total ordering by which we enforce acyclicity, and the slack bits used to limit the size of the parent sets. An arrow from a group of bits to a part of the Hamiltonian indicates that that part of the Hamiltonian is a function of that set of bits.
}
\end{figure*}
Putting together the parts of the Hamiltonian defined above, define
\begin{equation}
H(\mathbf d, \mathbf y, \mathbf r) \equiv  
H_{\text{score}}(\mathbf d) 
+ H_{\text{max}}(\mathbf d, \mathbf y) 
+ H_{\text{cycle}}(\mathbf d, \mathbf r).
\end{equation}
In the next section, we show lower bounds on the penalty weights therein that ensure that the ground state of the total Hamiltonian $H$ encodes the highest-scoring DAG with a maximum parent set size of $m$.
The sets of variables described above have the following sizes:
\begin{align}
|\{d_{ij}\}| & = n(n-1), \nonumber \\
|\{r_{ij}\}| & = \frac{n(n-1)}{2}, \text{and} \nonumber\\
|\{y_{il}\}| & = n \mu = n \left\lceil \log_2(m+1)\right\rceil.
\end{align}
Furthermore, while $H_{\text{max}}$ and $H_{\text{cycle}}$ are natively 2-local, $H_{\text{score}}$ is $m$-local. For each variable $x_i$ there are $\binom{n-1}{l}$ possible parent sets of size $l$ and the same number of corresponding $l$-local terms in $H_{\text{score}}$. If $m=3$, the full set of $\binom{n-1}{3}$ high-local terms $\left\{\prod_{j\in J}d_{ji} ||J|=3\right\}$ corresponding to parent sets of the variable $x_i$ can be reduced using $\lfloor \frac{(n-2)^2}{4}\rfloor$ ancilla variables. In total, $n\lfloor \frac{(n-2)^2}{4}\rfloor$ ancilla variables are needed to reduce $H_{\text{score}}$ to 2-local.

A quadratic pseudo-Boolean function can be identified with a graph whose vertices correspond to its arguments and whose edges correspond to non-zero quadratic terms. The graph identified with the Hamiltonian described above for $m=2$ has several features that indicate it may be difficult to embed in sparsely connected physical devices. First, for each variable $X_i$ there is a clique consisting of the variables $\{d_{ji}\} \cup \{y_{il}\}$, whose order is $(n-1) + \mu$. Second, the set of variables $\{r_{ij}\}$ are almost fully connected. 

\subsection{Utilizing Prior Information}
\label{sec:prior-info}
The mapping so far described assumes a uniform prior distribution over all possible DAGs of the appropriate size and with the given maximum number of parents. However, there are situations in which it may be desirable to fix the presence or absence of an arc in the search space. This could be because of domain knowledge or because hardware limitations prevent the implementation of the mapping for all arcs, in which case resort can be made to iterative search procedures such as the bootstrap method \cite{friedman1999}. To realize the reduction in qubits needed by accounting for such a reduced search space, suppose that we wish to only consider network structures that include the arc $(i, j)$, where $i < j$ without loss of generality. We then set $d_{ij}=1$, $d_{ji}=0$, and $r_{ij}=1$. Similarly, if $(i, j)$ is to be excluded, we set $d_{ij}=d_{ji}=0$ and keep $r_{ij}$ as a free variable. This can be done for any number of arcs. The Hamiltonian remains unchanged once these substitutions are made, and the lower bounds on the penalty weight remain sufficient, with the exception of the terms used in quadratization in the case $m>2$, in which case the quadratization should be done after substitution to utilize the reduction in degree of some terms.

\section{Penalty Weights}
\label{sec:penalty-weights}
In the expression above, there are several sets of free parameters called penalty weights: $\{\delta_{\text{max}}^{i}|1\leq i \leq n\}$, $\{\delta_{\text{consist}}^{ij}|1\leq i, j \leq n, i \neq j\}$, and $\{\delta_{\text{trans}}^{ijk}|1\leq i < j < k\}$. They are associated with penalty terms, i.e. parts of the Hamiltonian whose value is zero on states satisfying the corresponding constraint and is positive on states violating it. The purpose of their inclusion is to ensure that the energy-minimizing state of the total Hamiltonian satisfies the requisite constraints by increasing the energy of those that do not. More strongly, the penalty weights must be set such that the ground state of the total Hamiltonian is the lowest energy state of $H_{\text{score}}$ that satisfies the constraints. Here we provide sufficient lower bounds on the penalty weight necessary to ensure that this purpose is met. No claim is made to their necessity, and tighter lower bounds may exist.
    It is important to note that these bounds are mathematical, i.e. they ensure their purpose is met as stated above. In pure adiabatic quantum computation, in which the quantum system is in its ground state for the duration of the algorithm, this is sufficient (though the computation time necessary for the conditions of the adiabatic theorem to hold may be longer than otherwise if a penalized state has lower energy than the first excited unpenalized state). In practical quantum annealing, however, a combination of physical effects may cause the optimal value (in the sense of minimizing the energy of the lowest-energy state found, which may or may not be the global ground state) of the penalty weights to be less than these bounds. This remains the case even for bounds shown to be tight with respect to their mathematical properties.

The bound presented for each of the three sets of penalty weights is based on the notion that only the addition of an arc (i.e. changing some $d_{ij}$ from 0 to 1) can lead to the violation of two of the constrains we are concerned with: the maximum number of parents and the consistency of the arc bits and the order bits. Therefore, we can use a basis for how strongly the associated penalty needs to be the greatest difference in the energy of $H_{\text{score}}$ adding each arc can contribute. The penalty for the third constraint, the absence of directed 3-cycles among the order bits, will then be a function of the penalty for the consistency constraint.

Formally, we wish to set the penalties such that for any $\mathbf{d}$ violating at least one of the constraints, we have
\begin{equation}
\label{eq:constraint-inequality}
\min_{\mathbf{y}, \mathbf{r}} H(\mathbf{d}, \mathbf{y}, \mathbf{r}) > H_{\text{score}}(\mathbf{d}^*),
\end{equation}
where
\begin{equation}
\mathbf d^* \equiv \operatorname*{arg\,min}_{\substack{|\mathbf{d}'|\leq m \\ G(\mathbf{d}')\text{ is a DAG}}} H_{\text{score}}(\mathbf{d}').
\end{equation}
This is achieved by showing that for any such $\mathbf{d}$ violating at least one constraint, there is another $\mathbf{d}'$ that satisfies all the constraints such that
\begin{equation}
\min_{\mathbf{y}, \mathbf{r}}H(\mathbf{d}, \mathbf{y}, \mathbf{r}) \geq 
\min_{\mathbf{y}, \mathbf{r}}H(\mathbf{d}', \mathbf{y}, \mathbf{r}).
\end{equation}
Because $\mathbf{d}'$ satisfies all the constraints,
\begin{equation}
\min_{\mathbf{y}, \mathbf{r}}H(\mathbf{d}', \mathbf{y}, \mathbf{r}) =
H_{\text{score}}(\mathbf{d}'),
\end{equation}
which implies the inequality in \eqref{eq:constraint-inequality}.
In this section, we state the bounds and provide brief justification, but relegate the proofs to Appendix \ref{app:pen-bound-proofs}.
\subsection{Auxiliary Quantity}
In this section, we briefly define an auxiliary quantity,
\begin{restatable}{equation}{Deltadef}
\Delta_{ji}' = -\min_{\{d_{ki}|k\neq i,j\}} \left\lbrace\left.H_{\text{score}}^{(i)}\right|_{d_{ji}=1} - \left.H_{\text{score}}^{(i)}\right|_{d_{ji}=0}\right\rbrace,
\end{restatable}
that will allow us to define the maximum penalty weights associated with the bounds described previously.
For details of the calculation of this quantity, see Appendix \ref{app:delta-calc}.
In general it is possible that $\Delta_{ji}<0$ as defined above. We thus define the quantity 
\begin{equation}
\Delta_{ji} \equiv \max\{0, \Delta_{ji}'\}.
\end{equation}
In the proof of the bounds, the two following facts will be useful.

\begin{restatable}[Monotonicity of $H_{\textnormal{max}}$]{claim}{monotonicityhmax}
\label{claim:Hmax-monotonicity}

If $\mathbf{d}\geq \mathbf{d}'$, then $\min_{\mathbf{y}}H_{\textnormal{max}}(\mathbf{d}, \mathbf{y}) \geq \min_{\mathbf{y}}H_{\textnormal{max}}(\mathbf{d}', \mathbf{y})$.
\end{restatable}

\begin{restatable}[Monotonicity of $H_{\textnormal{cycle}}$]{claim}{monotonicityhcycle}
\label{Hcycle-monotonicity}

If $\mathbf{d}\geq \mathbf{d}'$, then 
$\min_{\mathbf{r}}H_{\textnormal{cycle}}(\mathbf{d}, \mathbf{r}) \geq 
\min_{\mathbf{r}}H_{\textnormal{cycle}}(\mathbf{d}', \mathbf{r})$.
\end{restatable}
These say simply that the removal of one or more arcs from $G(\mathbf d)$ cannot increase the values of $H_{\text{max}}$ nor $H_{\text{cycle}}$.

\subsection{``Maximum'' Penalty Weights}
\label{sec:max-penalty}
Here we show a lower bound for $\{\delta_{\text{max}}^{(i)}\}$ that guarantees that if $\mathbf{d}$ is such that $\max_{1\leq i \leq n}|\mathbf{d}_i|>m$ there exists a $\mathbf{d}'$ with lesser total energy such that $\max_{1\leq i\leq n}|\mathbf{d}_i|\leq m$. To do so, we show that if, for some $\mathbf{d}$ and $i$, $|\mathbf{d}_i|>m$, there is a $\mathbf{d}'$ such that $|\mathbf{d}_i'|=|\mathbf{d}_i|-1$ and 

\begin{equation}
\min_{\mathbf{y}, \mathbf{r}}H(\mathbf{d}, \mathbf{y}, \mathbf{r}) \geq 
\min_{\mathbf{y}, \mathbf{r}}H(\mathbf{d}', \mathbf{y}, \mathbf{r}).
\end{equation}
This idea can be applied iteratively to show that if, for some $\mathbf{d}$ and $i$, $|\mathbf{d}_i|>m$, there is some $\mathbf{d}'$ with lesser energy such that $\mathbf{d}_j'=\mathbf{d}_j$ for $j\neq i$ and $|\mathbf{d}_i'|\leq m$. This idea in turn can be applied iteratively to show that if for some $\mathbf{d}$ $\max_{1\leq i\leq n}|\mathbf{d}_i|>m$ there is a $\mathbf{d}'$ such that $\max_{1\leq i \leq n}|\mathbf{d}_i'|\leq m$.

\begin{restatable}{claim}{maxpenaltyone}
\label{claim:max-penalty-bound-1}

If $\delta_{\text{max}}^{(i)} > \max_{j\neq i} \Delta_{ji}$ for all $1\leq i \leq n$, then for all $\mathbf{d}$ such that, for some $i^*$, $|\mathbf{d}_{i^*}|>m$, there is a $\mathbf{d}'$ such that $|\mathbf{d}_{i^*}'|=|\mathbf{d}_{i^*}|-1$, $\mathbf d_i' = \mathbf d_i$ for all $i\neq i^*$, and $\min_{\mathbf{y}, \mathbf{r}}H(\mathbf{d}, \mathbf{y}, \mathbf{r}) > \min_{\mathbf{y}, \mathbf{r}}H(\mathbf{d}', \mathbf{y}, \mathbf{r})$.
\end{restatable}

\begin{restatable}[Sufficiency of ``Maximum'' Penalty Weight]{claim}{maxpenaltybound}
\label{claim:max-penalty-bound}

If $\delta_{\text{max}}^{(i)} > \max_{j\neq i} \Delta_{ji}$ for all $i$, then for all $\mathbf{d}$ such that $\max_{i}|\mathbf{d}_i|>m$, there is a $\mathbf{d}'\leq \mathbf{d}$ such that $\max_{i}|\mathbf{d}_i'|\leq m$ and $\min_{\mathbf{y}, \mathbf{r}}H(\mathbf{d}, \mathbf{y}, \mathbf{r}) > \min_{\mathbf{y}, \mathbf{r}}H(\mathbf{d}', \mathbf{y}, \mathbf{r})$.
\end{restatable}

\subsection{``Reduction'' Penalty Weights}
The degree of the ``score'' Hamiltonian $H_{\text{score}}$ is natively $m$-local as constructed. If $m=2$, as it often will be in practice, the total Hamiltonian is natively quadratic. If $m>2$, additional ancilla bits are needed to reduce the locality. The general method for doing this is to replace the conjunction of a pair bits with an ancilla bit and to add a penalty term with sufficiently strong weighting that penalizes states in which the ancillary bit is not equal to the conjunction to which it should be. For $m=3$, this can be done using $n\left\lfloor (n-2)^2 / 4\right\rfloor$ ancilla bits, but no more, where each $H_{\text{score}}^{(i)}$ containing $n-1$ arc bits is quadratized independently; furthermore, heuristic methods have been developed that reduce needed weight of the penalty terms \cite{babbush2014c}. For $m=4$, at most $n\binom{n-1}{2}$ ancilla bits are needed. More generally, $\mathcal O(n^{2\log d})$ is ancilla bits are needed \cite{boros2014}. Because the proof of the bounds on the other penalty weights are secular as to the degree of $H_{\text{score}}$ so long as $\{\Delta_{ij}\}$ is computed appropriately, the quadratization of $H_{\text{score}}$, including the addition of penalty terms and the needed weights therefor, can be done using the standard methods described in the literature independent of the other penalties described here.

\subsection{``Cycle'' Penalty Weights}
First, we that if the consistency penalty is set high enough, for any $\mathbf d$ encoding a graph with a 2-cycle, there is some $\mathbf d'$ encoding one whose minimal value of $H$ over $\mathbf y, \mathbf r$ is strictly less than that of $\mathbf d$.
\begin{restatable}[Removal of 2-cycles.]{claim}{removaltwocycle}
\label{claim:2-cycle-removal}
If $\delta_{\text{consist}}^{(ij)} > \max\{\Delta_{ij}, \Delta_{ji}\}$ for all $1\leq i < j \leq n$ , then for all $\mathbf d$ such that $G(\mathbf d)$ contains a 2-cycle, there is some $\mathbf d'\leq \mathbf d$ such that 
$G(\mathbf d')$ does not contain a 2-cycle and 
$\min_{\mathbf y, \mathbf r}H(\mathbf d, \mathbf y, \mathbf r) >
\min_{\mathbf y, \mathbf r}H(\mathbf d', \mathbf y, \mathbf r)$.
\end{restatable}
Second, we show that for any $\mathbf d$ that encodes a digraph without a 2-cycle, the minimal value of $H_{\text{consist}}$ over all $\mathbf r$ is zero.

\begin{restatable}[Sufficiency of ``Consistency'' Penalty Weights]{claim}{hconsistpenalty}
\label{claim:consistency-penalty-bound}

If $\delta_{\textnormal{consist}}^{(ij)}>(n-2)\max_{k\notin\{i, j\}}\delta_{\textnormal{trans}}^{(ijk)}$ for $1\leq i < j \leq n$, then for all $\mathbf{d}$ such that $G(\mathbf d)$ contains no 2-cycle, 
$H_{\textnormal{consist}}(\mathbf{d}, \mathbf{r}^*)=0$, where $\mathbf{r}^* = \operatorname*{arg\,min}_{\mathbf{r}}H_{\textnormal{cycle}}(\mathbf{d}, \mathbf{r})$.
\end{restatable}

Third, we show that for any $\mathbf d$ that encodes a digraph not containing a 2-cycle but that is not a DAG, there is some $\mathbf d'$ that does encode a DAG and whose minimal value of $H$ over all $\mathbf y, \mathbf r$ is strictly less than that of $\mathbf d$.

\begin{restatable}[Sufficiency of ``Transitivity'' Penalty Weights]{claim}{htranspenalty}
\label{claim:transitivity-penalty-bound}

If $\delta_{\textnormal{consist}}^{(ij)}>(n-2)\max_{k\notin\{i, j\}}\delta_{\textnormal{trans}}^{(ijk)}$ for $1\leq i < j \leq n$ and $\delta_{\textnormal{trans}}^{(ijk)}=\delta_{\textnormal{trans}}>\max_{\substack{1\leq i', j'\leq n\\i'\neq j'}}\Delta_{i'j'}$ for $1\leq i < j < k \leq n$, then for all $\mathbf d$ such that $G(\mathbf d)$ does not contain a 2-cycle but does contain a directed cycle there is some $\mathbf d'$ such that $G(\mathbf d')$ is a DAG and $\min_{\mathbf y, \mathbf r}H(\mathbf d, \mathbf y, \mathbf r) > \min_{\mathbf y, \mathbf r}H(\mathbf d', \mathbf y, \mathbf r)$.  
\end{restatable}

Lastly, we show that for all $\mathbf d$ that encode a digraph that is not a DAG, there is some $\mathbf d'$ that does encode a DAG and whose minimal value of $H$ over all $\mathbf y, \mathbf r$ is strictly less than that of $\mathbf d$.
\begin{restatable}[Sufficiency of ``Cycle'' Penalty Weights]{claim}{hcyclepenalty}
\label{claim:cycle-penalty-bound}
If $\delta_{\textnormal{consist}}^{(ij)}>(n-2)\max_{k\notin\{i,j\}}\delta_{\textnormal{trans}}^{(ijk)}$ for all $1\leq i < j \leq n$ and $\delta_{\textnormal{trans}}^{(ijk)}=\delta_{\textnormal{trans}}>\max_{\substack{1\leq i, j \leq n\\ i \neq j}}$ for all $1\leq i < j < k \leq n$, then for all $\mathbf{d}$ such that $G(\mathbf d)$ contains a directed cycle, there is a $\mathbf{d}'\leq \mathbf{d}$ such that $G(\mathbf{d}')$ is a DAG, and
$\min_{\mathbf{y}, \mathbf{r}}H(\mathbf{d}', \mathbf{y}, \mathbf{r}) < 
\min_{\mathbf{y}, \mathbf{r}}H(\mathbf{d}, \mathbf{y}, \mathbf{r})$.
\end{restatable}

\subsection{Overall Sufficiency}
Finally, we show that the digraph encoded in the ground state of the total Hamiltonian $H$ is a DAG and has a maximum parent set size of at most $m$, and that it is the solution to the corresponding \textsc{BNSL} instance.

\begin{restatable}[Overall Sufficiency]{claim}{overallpenalty}
\label{claim:total-bounds}
If 
$\delta_{\textnormal{max}}^{(i)} > \max_{j\neq i}\Delta_{ji}$
for all $1\leq i \leq n$, 
$\delta_{\textnormal{consist}}^{(ij)} > 
(n-2)\max_{k\notin\{i,j\}}\delta_{\textnormal{trans}}^{(ijk)}$ 
for all $1\leq i < j \leq n$ and 
$\delta_{\textnormal{trans}}^{(ijk)}
=\delta_{\textnormal{trans}}
> \max_{\substack{1\leq i',j'\leq n\\i'\neq j'}}\Delta_{i'j'}$ 
for all $1\leq i < j < k \leq n$, then 
$H(\mathbf{d}^*, \mathbf y, \mathbf r) = 
\min_{\substack{\max_i|\mathbf{d}_i|\leq m \\ G(\mathbf{d})\, \textnormal{is a DAG}}} H_{\textnormal{score}}(\mathbf{d}, \mathbf y, \mathbf r)$,
$G(\mathbf d^*)$ is a DAG, and
$\max_i|\mathbf d_i^*| \leq m$, where
$\mathbf d^* = 
\operatorname{arg\,min}_{\mathbf d}\left\{\min_{\mathbf y, \mathbf r}H(\mathbf d, \mathbf y, \mathbf r)\right\}$.
\end{restatable}

The strict inequalities used in the specification of the lower bounds ensures that the global ground state is a score-maximizing DAG with maximum parent set size $m$, but replacing them with weak inequalities is sufficient to ensure that the ground state energy is the greatest score over all DAGs with maximum parent set size $m$. However, the latter is of little interest in the present situation because it is the DAG itself that is of interest, not its score per se.

\section{Conclusion}
\label{sec:conclusion}
We have introduced a mapping from the native formulation of \textsc{BNSL} to
\textsc{QUBO} that enables the solution of the former using novel methods. 

The mapping is unique
amongst known mappings of optimization problems to \textsc{QUBO} in that the
logical structure is instance-independent for a given problem size. This
enables the expenditure of considerably more computational resources on the
problem of embedding the logical structure into a physical device because such
an embedding need only be done once and reused for new instances. The problem
addressed, \textsc{BNSL}, is special among optimization problems in that approximate solutions thereto often have value rivaling that of the exact solution. This property, along with the general intractability of exact solution, implies the great value of efficient heuristics such as SA or QA implemented using this mapping.

At present, only problems of up to seven BN variables can be embedded in
existing quantum annealing hardware (i.e. the D-Wave Two chip installed at NASA
Ames Research Center), whereas classical methods are able to deal with many of tens of BN variables.
Nevertheless, the quantum state of the art is quickly advancing, and it is
conceivable that quantum annealing could be competitively applied to
\textsc{BNSL} in the near future.
Given the already advanced state of classical simulated annealing code,
it is similarly conceivable that its application to the \textsc{QUBO} form
described here could be competitive with other classical methods for solving
\textsc{BNSL}.

\section{Acknowledgements}
This work was supported by the AFRL Information Directorate under grant F4HBKC4162G001.  All opinions, findings, conclusions, and recommendations expressed in this material are those of the authors and do not necessarily reflect the views of AFRL. The authors would also like to acknowledge support from the NASA Advanced Exploration Systems program and NASA Ames Research Center.
R. B. and A. A.-G. were supported by the National Science Foundation under award NSF CHE-1152291. The authors are grateful to David Tempel, Ole Mengshoel, and Eleanor Rieffel for useful discussions.

\appendix
\section*{Appendix}
\section{Calculation of $\Delta_{ij}$}
\label{app:delta-calc}
Recall the definition of the auxillary quantity,
\Deltadef*
Note that $H_{\text{score}}^{(i)}$ can be decomposed as
\begin{align}
H_{\text{score}}^{(i)} &=  
\sum_{\substack{J\subset\{1,\ldots,n\}\setminus\{i\} \\ |J|\leq m}}\left(w_i(J)\prod_{k\in J}d_{ki}\right) \nonumber\\
&= \sum_{\substack{J\subset\{1,\ldots,n\}\setminus\{i, j\} \\ |J|\leq m}}\left(w_i(J)\prod_{k\in J}d_{ki}\right) \nonumber\\
&\quad + \sum_{\substack{J\subset\{1,\ldots,n\}\setminus\{i,j\}\\|J|\leq m-1}}\left(w_i(J\cup \{j\})d_{ji}\prod_{k\in J}d_{ki}\right), 
\end{align}
where the first term is independent of $d_{ji}$ and thus cancels in the argument of the minimization on the right-hand side of Equation \ref{thmt@@Deltadef}. Thus $\Delta_{ji}$ simplifies to
\begin{align}
\label{eq:Delta-simp}
&{\Delta_{ji}'} \nonumber\\
&=-\min_{\{d_{ki}|k\neq i, j\}}\left\{ \sum_{\substack{J\subset\{1,\ldots,n\}\setminus\{i,j\}\\|J|\leq m-1}}\left(w_i(J\cup\{j\})\prod_{k\in J}d_{ki}\right)\right\} \nonumber\\
&= \max_{\{d_{ki}|k\neq i, j\}}\left\{ -\sum_{\substack{J\subset\{1,\ldots,n\}\setminus\{i,j\}\\|J|\leq m-1}}\left(w_i(J\cup\{j\})\prod_{k\in J}d_{ki}\right)\right\}.
\end{align}

For $m=1$, $\Delta_{ji}$ is trivially $-w_i(\{j\})$, the constant value of the expression to be extremized in Equation \ref{eq:Delta-simp} regardless of the values of $\{d_{ki}|k\neq i,j\}$. For $m=2$, $\Delta_{ji}$ can still be calculated exactly:
\begin{align}
&{\Delta_{ji}'} \nonumber\\
 &= \max_{\{d_{ki}|k\neq i,j\}}\left\{ -\sum_{\substack{J\subset\{1,\ldots,n\}\setminus\{i,j\}\\|J|\leq 1}}\left(w_i(J\cup\{j\})\prod_{k\in J}d_{ki}\right)\right\} \nonumber\\
&= \max_{\{d_{ki}|k\neq i, j\}}\left\{ -w_i(\{j\}) - \sum_{\substack{1\leq k \leq n\\ k\neq i, j}}d_{ki}\right\} \nonumber\\
&= - w_i(\{j\}) - \sum_{\substack{1\leq k \leq n \\ k \neq i, j \\ w_i\{j, k\}) < 0}}w_i(\{j, k\}) \nonumber\\
&= - w_i(\{j\}) - \sum_{\substack{1\leq k \leq n \\ k \neq i, j}}\min\{0, w_i(\{j, k\})\}.
\end{align}
However, for $m\geq 3$, calculation of the extremum in Equation \ref{eq:Delta-simp} is an intractable optimization problem in its own right and therefore we must resort to a reasonable bound. Because $\Delta_{ji}$ will be used in finding a lower bound on the necessary penalty weights, caution ditates that we use, if needed, a greater value than necessary. A reasonable upper bound on the true value is:
\begin{align}
&{\Delta_{ji}'} \nonumber\\ 
&= \max_{\{d_{ki}|k\neq i, j\}}\left\{ -\sum_{\substack{J\subset\{1,\ldots,n\}\setminus\{i,j\}\\|J|\leq m-1}}\left(w_i(J\cup\{j\})\prod_{k\in J}d_{ki}\right)\right\} \nonumber \\
&\leq  -\sum_{\substack{J\subset\{1,\ldots,n\}\setminus\{i,j\}\\|J|\leq m-1 \\ w_i(J\cup\{j\}) < 0}}w_i(J\cup\{j\}) \nonumber \\
&=  -\sum_{\substack{J\subset\{1,\ldots,n\}\setminus\{i,j\}\\|J|\leq m-1}}\min\{ 0, w_i(J\cup\{j\})\}.
\end{align}

\section{Proofs of Penalty Weight Lower Bounds}
\label{app:pen-bound-proofs}
\monotonicityhmax*
\begin{proof}
$\mathbf{d} \geq \mathbf{d}'$ implies $\mathbf{d}_i \geq \mathbf{d}_i'$ and thus $|\mathbf{d}_i|\geq |\mathbf{d}_i'|$ for $1\leq i \leq n$. By design, $\min_{\mathbf{y}_i}H_{\text{max}}(\mathbf{d}_i, \mathbf{y}_i) = \delta_{\text{max}}^{(i)}\max\{ 0, |\mathbf{d}_i| - m\}$. 
Let $\mathbf{y}^* \equiv \operatorname*{arg\,min}_{\mathbf{y}}H(\mathbf{d}, \mathbf{y})$. Then
\begin{align}
\min_{\mathbf{y}}H_{\text{max}}(\mathbf{d}, \mathbf{y}) &= 
\sum_{i=1}^n \min_{\mathbf{y}_i} H_{\text{max}}(\mathbf{d}_i, \mathbf{y}_i)  \nonumber\\
&= \sum_{i=1}^n \delta_{\text{max}}^{(i)}\max\{ 0, |\mathbf{d}_i| - m\} \nonumber\\
&\geq \sum_{i=1}^n \delta_{\text{max}}^{(i)}\max\{ 0, |\mathbf{d}_i'| - m\} \nonumber\\
&= \sum_{i=1}^n \min_{\mathbf{y}_i} H_{\text{max}}(\mathbf{d}_i', \mathbf{y}_i)  \nonumber\\
&= \min_{\mathbf{y}}H_{\text{max}}(\mathbf{d}', \mathbf{y}).
\end{align}
\end{proof}

\monotonicityhcycle*
\begin{proof}
In the statement of the claim, we implicitly assume that $\delta_{\text{consist}}^{(ij)}>0$ for all $1\leq i < j \leq n$. Let $\mathbf{r}^*=\operatorname*{arg\,min}_{\mathbf{r}} H_{\textnormal{cycle}}(\mathbf{d}, \mathbf{r})$. Because $d_{ji}\geq d_{ji}'$ for all $i, j$ such that$i\neq j$ and $1\leq i, j \leq n$, and because $0\leq r_{ij} \leq 1$ for all $1\leq i < j \leq n$,
\begin{align}
&{\min_{\mathbf{r}}H_{\text{cycle}}(\mathbf{d}, \mathbf{r})} \nonumber\\
&=\min_{\mathbf{r}}\left\{ H_{\textnormal{consist}}(\mathbf{d}, \mathbf{r}) + H_{\textnormal{trans}}(\mathbf{r})\right\} \nonumber\\
&= H_{\textnormal{consist}}(\mathbf{d}, \mathbf{r}^*) + H_{\textnormal{trans}}(\mathbf{r}^*) \nonumber\\
&= \left(\sum_{1\leq i < j \leq n} \delta_{\textnormal{consist}}^{(ij)} \left[d_{ji}r_{ij}^* + d_{ij}(1-r_{ij}^*)\right]\right) + H_{\text{trans}}(\mathbf{r}^*) \nonumber\\
&\geq \left(\sum_{1\leq i < j \leq n} \delta_{\textnormal{consist}}^{(ij)} \left[d_{ji}'r_{ij}^* + d_{ij}'(1-r_{ij}^*)\right]\right) + H_{\text{trans}}(\mathbf{r}^*) \nonumber\\
&=H_{\textnormal{consist}}(\mathbf{d}',\mathbf{r}^*) + H_{\textnormal{trans}}(\mathbf{r}^*) \nonumber\\
&\geq \min_{\mathbf{r}} \left[H_{\text{consist}}(\mathbf{d}', \mathbf{r}) + H_{\text{trans}}(\mathbf{r})\right] \nonumber\\
&=\min_{\mathbf{r}}H_{\text{cycle}}(\mathbf{d}', \mathbf{r})
\end{align}

\end{proof}

\maxpenaltyone*
\begin{proof}
We prove the existence of such a $\mathbf{d}'$ by construction. 
Let $\mathbf{d}_i'\equiv \mathbf{d}_i$ for all $i\neq i^*$. 
Let $\mathbf{d}_{i^*}' \equiv \left.\mathbf{d}_{i^*}\right|_{d_{j^*i^*}=0}$, where $j^*=\operatorname*{arg\,min}_{j\in \{j|d_{ji^*}=1\}}\Delta_{ji^*}$.
First, we note that by design $
\min_{\mathbf y_i} H_{\text{max}}^{(i)}(\mathbf d_i, \mathbf y_i)=\min\{0, \delta_{\text{max}}^{(i)}(|\mathbf d_i| - m)\}$. Thus
\begin{align}
&{\min_{\mathbf{y}_{i^*}}H_{\text{max}}^{(i^*)}(\mathbf{d}_{i^*}, \mathbf{y}_{i^*}) - \min_{\mathbf{y}_{i^*}}H_{\text{max}}^{(i^*)}(\mathbf{d}_{i^*}', \mathbf{y}_{i^*})} \nonumber\\ 
&= \left[\delta_{\text{max}}^{(i^*)}(|\mathbf{d}|-m)\right] + \left[\delta_{\text{max}}^{(i^*)}(|\mathbf{d}'|-m)\right] \nonumber\\
&= \delta_{\text{max}}^{(i^*)}(|\mathbf{d}|-|\mathbf{d}'|)  \nonumber\\
&= \delta_{\text{max}}^{(i^*)}  \nonumber\\
&> \max_{j\neq i^*}\Delta_{ji^*}  \nonumber\\
&\geq \Delta_{j^*i^*}  \nonumber\\
&\geq H_{\text{score}}^{(i^*)}(\mathbf{d}_{i^*}') - H_{\text{score}}^{(i^*)}(\mathbf{d}_{i^*}). 
\end{align}
which rearranges to
\begin{align}
&{H_{\text{score}}^{(i^*)}(\mathbf{d}_{i^*}) + \min_{\mathbf{y}_{i^*}}H_{\text{max}}^{(i^*)}(\mathbf{d}_{i^*}, \mathbf{y}_{i^*})} \nonumber\\
&> H_{\text{score}}^{(i^*)}(\mathbf{d}_{i^*}') + \min_{\mathbf{y}_{i^*}}H_{\text{max}}^{(i^*)}(\mathbf{d}_{i^*}', \mathbf{y}_{i^*}).
\end{align}

\noindent In context,
\begin{align}
&{\min_{\mathbf{y}}\left[H_{\text{score}}(\mathbf{d}) + H_{\text{max}}(\mathbf{d}, \mathbf{y})\right]} \nonumber\\
 &= H_{\text{score}}(\mathbf{d}) + \min_{\mathbf{y} }H_{\text{max}}(\mathbf{d}, \mathbf{y}) \nonumber\\
&=\sum_{i\neq i^*}\left[H_{\text{score}}^{(i)}(\mathbf{d}_i) + \min_{\mathbf{y}_i}H_{\text{max}}^{(i)}(\mathbf{d}_i, \mathbf{y}_i)\right] \nonumber\\
&\quad+ \left[H_{\text{score}}^{(i^*)}(\mathbf{d}_{i^*}) + \min_{\mathbf{y}_{i^*}}H_{\text{max}}^{(i^*)}(\mathbf{d}_{i^*}, \mathbf{y}_{i^*})\right]  \nonumber\\
&=\sum_{i\neq i^*}\left[H_{\text{score}}^{(i)}(\mathbf{d}_i') + \min_{\mathbf{y}_i}H_{\text{max}}^{(i)}(\mathbf{d}_i', \mathbf{y}_i)\right] \nonumber\\
&\quad+ \left[H_{\text{score}}^{(i^*)}(\mathbf{d}_{i^*}) + \min_{\mathbf{y}_{i^*}}H_{\text{max}}^{(i^*)}(\mathbf{d}_{i^*}, \mathbf{y}_{i^*})\right]  \nonumber\\
&> \sum_{i\neq i^*}\left[H_{\text{score}}^{(i)}(\mathbf{d}_i') + \min_{\mathbf{y}_i}H_{\text{max}}^{(i)}(\mathbf{d}_i', \mathbf{y}_i)\right] \nonumber\\
&\quad+ \left[H_{\text{score}}^{(i^*)}(\mathbf{d}_{i^*}') + \min_{\mathbf{y}_{i^*}}H_{\text{max}}^{(i^*)}(\mathbf{d}_{i^*}', \mathbf{y}_{i^*})\right] \nonumber \\
&= H_{\text{score}}(\mathbf{d}') + \min_{\mathbf{y} }H_{\text{max}}(\mathbf{d}', \mathbf{y}) \nonumber\\
&=\min_{\mathbf{y}}\left[H_{\text{score}}(\mathbf{d}') + H_{\text{max}}(\mathbf{d}', \mathbf{y})\right].
\end{align}

By Claim \ref{claim:Hmax-monotonicity} and the fact that $\mathbf d' \leq
\mathbf d$,
\begin{equation}
\min_{\mathbf r}H_{\text{max}}(\mathbf d, \mathbf r) \geq 
\min_{\mathbf r}H_{\text{max}}(\mathbf d', \mathbf r),
\end{equation}
and so
\begin{align}
&{\min_{\mathbf y, \mathbf r}H(\mathbf d, \mathbf y, \mathbf r)}  \nonumber\\
&= H_{\text{score}}(\mathbf d) + \min_{\mathbf y}H_{\text{max}}(\mathbf d, \mathbf y) +
\min_{\mathbf r}H_{\text{cycle}}(\mathbf d, \mathbf r)  \nonumber\\
&> H_{\text{score}}(\mathbf d') + \min_{\mathbf y}H_{\text{max}}(\mathbf d', \mathbf y) +
\min_{\mathbf r}H_{\text{cycle}}(\mathbf d', \mathbf r)  \nonumber\\
&= \min_{\mathbf y, \mathbf r}H(\mathbf d, \mathbf y, \mathbf r).
\end{align}
 
\end{proof}

\maxpenaltybound*
\begin{proof}
We prove the sufficiency of the given bound by iterative application of Claim \ref{claim:max-penalty-bound-1}.
Let $\mathbf d^{(0, 0)} \equiv \mathbf d$. For all $i$, if $|\mathbf{d}_i|>m$, let 
$\mathbf{d}^{(i, |\mathbf{d}_i|-m)} \equiv \mathbf{d}^{(i-1, 0)}$, and if
$|\mathbf d_i| \leq m$, let $\mathbf d^{(i, 0)}\equiv \mathbf d^{(i-1, 0)}$.
\\
For all $1\leq i \leq n$ and $x$ such that $1\leq i \leq n$ and $1\leq x \leq \max\{ 0, |\mathbf{d}_i| - m \}\}$, $|\mathbf{d}_i^{(i,x)}|>m$ and so by Claim \ref{claim:max-penalty-bound-1} there is a $\mathbf{d}^{(i,x-1)}\leq \mathbf{d}^{(i,x)}$ such that 
$|\mathbf{d}^{(i,x-1)}| = |\mathbf{d}^{(i,x)}| - 1$ 
and 
$\min_{\mathbf{y}, \mathbf{r}}H(\mathbf{d}^{(i, x-1)}, \mathbf{y}, \mathbf{r}) 
< \min_{\mathbf{y}, \mathbf{r}}H(\mathbf{d}^{(i, x)}, \mathbf{y}, \mathbf{r})$.
Then for all $i$, if $|\mathbf{d}_i|>m$, there is a sequence 
$\mathbf d^{(i, |\mathbf d_i| - m)}, \ldots, \mathbf d^{(i, x)}, \ldots, \mathbf d^{(i, 0)}$ 
such that 
$\mathbf{d}^{(i,0)}\leq \mathbf{d}^{(i,|\mathbf{d}_i|-m)}$ 
and 
$\min_{\mathbf{y}, \mathbf{r}}H(\mathbf{d}^{(i,0)}, \mathbf{y}, \mathbf{r}) 
< \min_{\mathbf{y}, \mathbf{r}}H(\mathbf{d}^{(i, |\mathbf{d}_i|-m)}, \mathbf{y}, \mathbf{r})$. 
Similarly, for all $i$, 
$\mathbf d^{(i, 0)} \leq \mathbf d^{(i-1, 0)}$ and
$\min_{\mathbf y, \mathbf r}H(\mathbf d^{(i, 0)} 
\leq \min_{\mathbf y, \mathbf r}H(\mathbf d^{(i-1)}, \mathbf y, \mathbf r)$,
with strict inequality if $|\mathbf d_i|>0$ and equality if $|\mathbf d_i|=0$. 
Thus there is a sequence $\mathbf d^{(0, 0)}, \mathbf d^{(1, 0)}, \ldots, \mathbf d^{(i, 0)}, \ldots, \mathbf d^{(n, 0)}$ such that
$\mathbf{d}^{(n,0)}\leq \mathbf d^{(0, 0)} =\mathbf{d}$, 
$\max_{i}|\mathbf{d}_i^{(n,0)}|\leq m$, and 
$\min_{\mathbf{y}, \mathbf{r}}H(\mathbf{d}', \mathbf{y}, \mathbf{r})
< \min_{\mathbf{y}, \mathbf{r}}H(\mathbf{d}, \mathbf{y}, \mathbf{r})$. 
Setting $\mathbf{d}'\equiv \mathbf{d}^{(n,0)}$ completes the proof.
\end{proof}

\removaltwocycle*
\begin{proof}
Let $\mathbf d^{(0)}\equiv \mathbf d$ and $l^*$ be the number of 2-cycles contained in $G(\mathbf d)$. 
The claim is proved iteratively by showing that for all $\mathbf d^{(l)}$ such that $G(\mathbf d^{(l)})$ contains a 2-cycle there exists some $\mathbf d^{(l + 1)}$ such that $\min_{\mathbf y, \mathbf r}H(\mathbf d^{(l)}, \mathbf y, \mathbf r) > \min_{\mathbf y, \mathbf r}(\mathbf d^{(l + 1)}, \mathbf y, \mathbf r)$ and $G(\mathbf d^{(l + 1)})$ contains one fewer 2-cycle. 
Because a graph of fixed order can only have a finite number of 2-cycles, this implies the existence of a sequence $\mathbf d, \mathbf d^{(1)}, \ldots, \mathbf d^{(l)}, \ldots, \mathbf d^{(l^*)}$ such that $\mathbf d^{(l^*)}$ meets the desiderata.

Consider an arbitrary $\mathbf d^{(l)}$.
If $G(\mathbf d^{(l)})$ does not contain a directed 2-cycle, then $l=l^*$ and so we set $\mathbf d'=\mathbf d^{(l^*)}$ to complete the proof.
Otherwise, choose some 2-cycle in $G(\mathbf d^{(l)})$ arbitrarily, i.e. some pair $\{i, j\}$ such that $(i, j), (j, i) \in E(G(\mathbf d))$, or, equivalently, that $d_{ij}=d_{ji}=1$. Without loss of generality, assume $i<j$. Let $\mathbf r^*\equiv \operatorname{arg\, min}_{\mathbf r}H_{\text{cycle}}(\mathbf d^{(l)}, \mathbf r)$ and 
\begin{equation}
(i^*, j^*) \equiv \begin{cases} (j, i), &r_{ij}^*=1, \\ (i, j), &r_{ij}^*=0, \end{cases}
\end{equation}
i.e. the arc in $G(\mathbf d^{(l)}$ inconsistent with $G(\mathbf r^*)$.
Define $\mathbf d^{(l + 1)}$ such that $d_{ij}^{(l+1)} = \begin{cases} d_{ij}, &(i, j) \neq (i^*, j^*) \\ 0=1- d_{ij}, &(i, j) = (i^*, j^*) \end{cases}$. Thus $\mathbf d^{(l+1)} \leq \mathbf d^{(l)}$ and $|\mathbf d^{(l + 1)}| = |\mathbf d^{(l)}| -1$. By construction, $G(\mathbf d^{(l+1)})$ contains one fewer 2-cycle than $G(\mathbf d^{(l)})$.  Furthermore,
\allowdisplaybreaks\begin{align}
&{\min_{\mathbf{r}}H_{\text{cycle}}(\mathbf{d}^{(l)}, \mathbf{r}) -\min_{\mathbf{r}}H_{\text{cycle}}(\mathbf{d}^{(l + 1)}, \mathbf{r})} \nonumber\\
&= H_{\text{cycle}}(\mathbf{d}^{(l)}, \mathbf{r}^*) -\min_{\mathbf{r}}H_{\text{cycle}}(\mathbf{d}^{(l + 1)}, \mathbf{r})\nonumber \\
&\geq H_{\text{cycle}}(\mathbf{d}^{(l)}, \mathbf{r}^*) - H_{\text{cycle}}(\mathbf{d}^{(l + 1)}, \mathbf{r}^*) \nonumber\\
&= \Bigg[\Bigg(\sum_{1\leq i < j \leq n} \delta_{\textnormal{consist}}^{(ij)} \left[d_{ji}^{(l)}r_{ij}^* + d_{ij}^{(l)}(1-r_{ij}^*)\right]\Bigg) \nonumber\\
&\quad + H_{\text{trans}}(\mathbf{r}^*)\Bigg] \nonumber\\
&\quad-\Bigg[ \Bigg(\sum_{1\leq i < j \leq n} \delta_{\textnormal{consist}}^{(ij)} \left[d_{ji}^{(l+1)}r_{ij}^* + d_{ij}^{(l+1)}(1-r_{ij}^*)\right]\Bigg) \nonumber\\
&\quad + H_{\text{trans}}(\mathbf{r}^*)\Bigg] \nonumber \\
&=\left(\sum_{1\leq i < j \leq n} \delta_{\textnormal{consist}}^{(ij)} \left[d_{ji}^{(l)}r_{ij}^* + d_{ij}^{(l)}(1-r_{ij}^*)\right]\right) \nonumber\\
& \quad -\left(\sum_{1\leq i < j \leq n} \delta_{\textnormal{consist}}^{(ij)} \left[d_{ji}^{(l+1)}r_{ij}^* + d_{ij}^{(l+1)}(1-r_{ij}^*)\right]\right)  \nonumber\\
&=\sum_{1\leq i < j \leq n} \delta_{\textnormal{consist}}^{(ij)} \left[(d_{ji}^{(l)}-d_{ji}^{(l+1)})r_{ij}^* \right.\nonumber\\
&\hspace{90pt} \left.+ (d_{ij}^{(l)}-d_{ij}^{(l+1)})(1-r_{ij}^*)\right]
\nonumber\\
&=
\begin{cases} \delta_{\text{consist}}^{(j^*i^*)} (d_{i^*j^*}^{(l)} - d_{i^*j^*}^{(l+1)})r_{j^*i^*}^*, & j^* < i^* \\
\delta_{\text{consist}}^{(i^*j^*)}(d_{i^*j^*}^{(l)} - d_{i^*j^*}^{(l+1)})(1-r_{i^*j^*}^*), & i^* < j^* \end{cases} \nonumber\\
&=
\begin{cases} \delta_{\text{consist}}^{(j^*i^*)}, & j^* < i^* \\
\delta_{\text{consist}}^{(i^*j^*)}, & i^* < j^* \end{cases} \nonumber\\
&>\Delta_{i^*j^*} \nonumber \\
&\geq H_{\text{score}}(\mathbf{d}^{(l+1)}) - H_{\text{score}}(\mathbf{d}^{(l)}).
\end{align}
By Claim \ref{claim:Hmax-monotonicity} and the fact that $\mathbf d^{(l+1)} \leq \mathbf d^{(l)}$, 
\begin{equation}
\min_{\mathbf y}H_{\text{max}}(\mathbf d^{(l)}) \geq \min_{\mathbf y}H_{\text{max}}(\mathbf d^{(l+1)}).
\end{equation}
Thus, 
\begin{align}
&{\min_{\mathbf y, \mathbf r}H(\mathbf d^{(l)}, \mathbf y, \mathbf r) - \min_{\mathbf y, \mathbf r}H(\mathbf d^{(l+1)}, \mathbf y, \mathbf r)} \nonumber\\
& =\Big(H_{\text{score}}(\mathbf d^{(l)}) + \min_{\mathbf r} H_{\text{cycle}}(\mathbf d^{(l)}, \mathbf r) \nonumber\\
&\qquad+ \min_{\mathbf y} H_{\text{max}}(\mathbf d^{(l)}, \mathbf y)\Big) \nonumber\\
&\qquad - \Big(H_{\text{score}}(\mathbf d^{(l+1)}) + \min_{\mathbf r} H_{\text{cycle}}(\mathbf d^{(l+1)}, \mathbf r) \nonumber\\
&\qquad \qquad + \min_{\mathbf y} H_{\text{max}}(\mathbf d^{(l+1)}, \mathbf y)\Big) \nonumber\\
& \geq \left(H_{\text{score}}(\mathbf d^{(l)}) + \min_{\mathbf r} H_{\text{cycle}}(\mathbf d^{(l)}, \mathbf r)\right) \nonumber\\
&\quad - \left(H_{\text{score}}(\mathbf d^{(l+1)}) + \min_{\mathbf r} H_{\text{cycle}}(\mathbf d^{(l+1)}, \mathbf r)\right) \nonumber\\
&> 0,
\end{align}
which rearranges to the desired inequality. \end{proof}

\hconsistpenalty*
\begin{proof}
We prove the claim via its contrapositive: for all $\mathbf{d}, \mathbf{r}$, if $H_{\text{consist}}(\mathbf{d}, \mathbf{r}) > 0$, there is some $\mathbf{r}'$ such that 
$H_{\text{cycle}}(\mathbf{d}, \mathbf{r}) > H_{\text{cycle}}(\mathbf{d}, \mathbf{r}')$, so $\mathbf{r}\neq \operatorname*{arg\,min}_{\mathbf{r}}H_{\text{cycle}}(\mathbf{d}, \mathbf{r})$.

Consider an arbitrary $\mathbf d$ and some $\mathbf r$ such that $H_{\text{consist}}(\mathbf d, \mathbf r) > 0$. The positivity of $H_{\text{consist}}(\mathbf d, \mathbf r)$ indicates that there is at least one inconsistency between $\mathbf d$ and $\mathbf r$, i.e. there is some $(i^*, j^*)$ such that 
$d_{i^*j^*}=\begin{cases} r_{j^*i^*}, & i^* > j^*\\ 1 - r_{i^*j^*}, & i^* < j^* \end{cases}=1$. For convenience, we prove the claim for the case in which $i^* < j^*$; the proof provided can be easily modified for the case in which $i^* > j^*$. Let $\mathbf r'$ be the same as $\mathbf r$ exept in the bit corresponding to this inconsistency: 
$r_{ij}' \equiv \begin{cases} r_{ij} & (i, j) \neq (i^*, j^*) \\ 1 - r_{ij}, & (i, j) = (i^*, j^*) \end{cases}$.
Then
\allowdisplaybreaks\begin{align}
&H_{\text{consist}}(\mathbf d, \mathbf r) - H_{\text{consist}}(\mathbf d, \mathbf r') \nonumber\\
&= \sum_{1\leq i < j \leq n} \left[H_{\text{consist}}^{(ij)}(d_{ij}, d_{ji}, r_{ij}) \right. \nonumber\\
&\hspace{60pt} \left.- H_{\text{consist}}^{(ij)}(d_{ij}, d_{ji}, r_{ij}')\right] \nonumber\\
&= H_{\text{consist}}^{(i^*j^*)}(d_{i^*j^*}, d_{j^*i^*}, r_{i^*j^*}) \nonumber\\
&\qquad -H_{\text{consist}}^{(i^*j^*)}(d_{i^*j^*}, d_{j^*i^*}, r_{i^*j^*}')\\
&=\delta_{\text{consist}}^{(i^*j^*)}\left[d_{j^*i^*} r_{i^*j^*} + d_{i^*j^*}(1-r_{i^*j^*})\right] \nonumber \\
&\quad - \delta_{\text{consist}}^{(i^*j^*)} \left[d_{j^*i^*} r_{i^*j^*}' + d_{i^*j^*}(1-r_{i^*j^*}')\right] \nonumber \\
&=\delta_{\text{consist}}^{(i^*j^*)}\left[-d_{j^*i^*} + d_{i^*j^*}\right]
\nonumber \\
&=\delta_{\text{consist}}^{(i^*j^*)}.
\end{align}
Furthermore,
\begin{align}
&{H_{\text{trans}}(\mathbf r) - H_{\text{trans}}(\mathbf r')} \nonumber\\
&= \sum_{1\leq i < j < k \leq n}H_{\text{trans}}^{(ijk)}(r_{ij}, r_{ik}, r_{jk})\nonumber\\
&\qquad- \sum_{1\leq i < j < k \leq n}H_{\text{trans}}^{(ijk)}(r_{ij}', r_{ik}', r_{jk}') \nonumber\\
&=\sum_{\substack{1\leq i < j < k \leq n \\ \{i^*, j^*\} \subset \{i, j ,k\}}}
\Big[H_{\text{trans}}^{(ijk)}(r_{ij}, r_{ik}, r_{jk}) \\
&\hspace{90pt}- H_{\text{trans}}^{(ijk)}(r_{ij}', r_{ik}', r_{jk}') \Big] \nonumber \\
&=\sum_{k<i^*} \delta_{\text{trans}}^{(ki^*j^*)}\big[
\left(r_{kj^*}+r_{ki^*}r_{i^*j^*} - r_{ki^*}r_{kj^*} - r_{i^*j^*}r_{kj^*} \right)\nonumber\\
&\quad - \left(r_{kj^*}'+r_{ki^*}'r_{i^*j^*}' - r_{ki^*}'r_{kj^*}' - r_{i^*j^*}'r_{kj^*}' \right)\big] \nonumber\\
&\quad+ \sum_{i^*<k<j^*} \delta_{\text{trans}}^{(i^*kj^*)} \big[
\left(r_{i^*j^*}+r_{i^*k}r_{kj^*} \right.\\
&\hspace{100pt} \left.- r_{i^*k}r_{i^*j^*} - r_{kj^*}r_{i^*j^*} \right) \nonumber\\
&\quad - \left(r_{i^*j^*}'+r_{i^*k}'r_{kj^*}' - r_{i^*k}'r_{i^*j^*}' - r_{kj^*}'r_{i^*j^*}' \right)\big] \nonumber\\
&\quad+ \sum_{j^*<k} \delta_{\text{trans}}^{(i^*j^*k)} \big[
\left(r_{i^*k}+r_{i^*j^*}r_{j^*k} \right. \nonumber\\
&\hspace{100pt} \left.- r_{i^*j^*}r_{i^*k} - r_{j^*k}r_{i^*k} \right) \nonumber\\
&\quad - \left(r_{i^*k}'+r_{i^*j^*}'r_{j^*k}' - r_{i^*j^*}'r_{i^*k}' - r_{j^*k}'r_{i^*k}' \right) \big] \nonumber\\
&=\sum_{k<i^*} \delta_{\text{trans}}^{(ki^*j^*)}
\left(-r_{ki^*} + r_{kj^*} \right)\nonumber\\
&\quad+ \sum_{i^*<k<j^*} \delta_{\text{trans}}^{(i^*kj^*)} 
\left(-1 + r_{i^*k} + r_{kj^*}\right) \nonumber\\
&\quad+ \sum_{j^*<k} \delta_{\text{trans}}^{(i^*j^*k)} 
\left(-r_{j^*k} + r_{i^*k}\right) \nonumber\\
&\leq \sum_{k<i^*} \delta_{\text{trans}}^{(ki^*j^*)}
+ \sum_{i^*<k<j^*} \delta_{\text{trans}}^{(i^*kj^*)} 
+ \sum_{j^*<k} \delta_{\text{trans}}^{(i^*j^*k)} \nonumber \\
&\leq (n-2)\max_{k\notin \{i^*, j^*\}}\delta_{\text{trans}}^{(i^*j^*k)}.
\end{align}

Together, the above imply
\begin{align}
&{H_{\text{cycle}}(\mathbf d, \mathbf r) -H_{\text{cycle}}(\mathbf d, \mathbf r')} \nonumber \\
&=H_{\text{consist}}(\mathbf d, \mathbf r) - H_{\text{consist}}(\mathbf d, \mathbf r') 
+ H_{\text{trans}}(\mathbf r) - H_{\text{trans}}(\mathbf r') \nonumber\\
&\geq \delta_{\text{consist}}^{(i^*j^*)} - (n-2)\max_{k\notin \{i, j\}}\delta_{\text{trans}}^{(ijk)} \nonumber\\
&>0.
\end{align}

\end{proof}

\htranspenalty*
\begin{proof}
Consider an arbitrary $\mathbf d^{(l)}$ such that $G(\mathbf d^{(l)})$ does not
contain a 2-cycle but does contain a directed cycle. Let $\mathbf r^{(l)}\equiv
\operatorname{arg\,min}_{\mathbf r}H_{\text{cycle}}(\mathbf d^{(l)}, \mathbf
r)$. By Claim \ref{claim:consistency-penalty-bound},
$H_{\text{consist}}(\mathbf d^{(l)}, \mathbf r^{(l)})=0$ and so
$H_{\text{cycle}}(\mathbf d^{(l)}, \mathbf r^{(l)}) = 
H_{\text{trans}}(\mathbf d^{(l)}, \mathbf r^{(l)})$. 

If $\delta_{\text{trans}}^{(ijk)}=\delta_{\text{trans}}$ 
for $1\leq i < j < k \leq n$, 
i.e. the trasitivity penalty weight is uniform for all directed triangles, then $H_{\text{trans}}(\mathbf r)$ is equal to the product of $\delta_{\text{trans}}$ and the number of directed triangles in the tournament $G(\mathbf r)$ for all $\mathbf r$. 

In any tournament with a positive number of directed triangles, there is always some arc whose switch of direction lowers the number of directed triangles. 
Let $(i^*, j^*)$ be such such an arc for $\mathbf r^{(l)}$. 
Define $\tilde{\mathbf r}^{(l)}$ such that 
$\tilde r_{ij}^{(l)} = \begin{cases} 
r_{ij}^{(l)},& (i,j) \neq (i^*,j^*)\\ 
1-r_{ij}^{(l)},& (i,j) = (i^*,j^*) \end{cases} $.
By construction, 
$H_{\text{trans}}(\mathbf r_l^*) - 
 H_{\text{trans}}(\tilde{\mathbf r}^{(l)}) \geq 
\delta_{\text{trans}}$.

It must be the case that $d_{i^*j^*}^{(l)}=1$. Suppose otherwise. Define some
$\mathbf r'$ such that $r_{ij}'=\begin{cases}
r_{ij}, &(i, j) \neq (i^*, j^*), \\
1-r_{ij}, &(i, j) = (i^*, j^*),\end{cases}$, which would have the properties
that
$H_{\text{consist}}(\mathbf d^{(l)}, \mathbf r')=0$ 
and, by construction,
$H_{\text{trans}}(\mathbf r') < H_{\text{trans}}(\mathbf r^{(l)}$, so that
$H_{\text{cycle}}(\mathbf d^{(l)}, \mathbf r') 
< H_{\text{cycle}}(\mathbf d^{(l)}, \mathbf r^{(l)}) 
\neq \min_{\mathbf r} H_{\text{cycle}}(\mathbf d^{(l)}, \mathbf r)$. Because
$G(\mathbf d^{(l)})$ does not contain a 2-cycle, $d_{j^*i^*}^{(l)}=0$.

Now, define $\mathbf d^{(l+1)}$ such that 
$d_{ij}^{(l+1)}=\begin{cases} d_{ij}^{(l)}, &(i, j) \neq (i^*, j^*),\\
0=1-d_{ij}^{(l)}, & (i, j) = (i^*, j^*).\end{cases}$
Then
\begin{align}
&{H_{\text{consist}}(\mathbf d^{(l+1)}, \tilde{\mathbf r}^{(l)})} \nonumber \\
&=\sum_{\substack{1\leq i < j \leq n \\ (i,j)\neq (i^*, j^*)}}
H_{\text{consist}}^{(ij)}(d_{ij}^{(l+1)}, d_{ji}^{(l+1)}, \tilde r_{ij}^{(l)}) \nonumber\\
&\qquad + H_{\text{consist}}^{(i^*j^*)}(d_{i^*j^*}^{(l+1)}, d_{j^*i^*}^{(l+1)}, \tilde r_{i*j*}^{(l)})  \nonumber\\
&=\sum_{\substack{1\leq i < j \leq n \\ (i,j)\neq (i^*, j^*)}}
H_{\text{consist}}^{(ij)}(d_{ij}^{(l)}, d_{ji}^{(l)}, r_{ij}^{(l)}) \nonumber\\
&\qquad + H_{\text{consist}}^{(i^*j^*)}(0, 0, r_{i*j*}^{(l)})  \nonumber\\
&\leq H_{\text{consist}}(\mathbf d^{(l)}, \mathbf r^{(l)})  \nonumber\\
&= 0
\end{align}
Because of this, it must be that 
$H_{\text{trans}}(\tilde{\mathbf r}^{(l)}) 
\geq H_{\text{trans}}(\mathbf r^{(l+1)})$, 
whose negation contradicts the definition of $\mathbf r^{(l+1)}$.
Therefore,
\begin{align}
&{H_{\text{trans}}(\mathbf r^{(l)}) - H_{\text{trans}}(\mathbf r^{(l+1)})} \nonumber\\
&\geq H_{\text{trans}}(\mathbf r^{(l)}) - H_{\text{trans}}(\tilde{\mathbf
r}^{(l)})
\nonumber\\
&\geq \delta_{\text{trans}} \nonumber\\
&>\Delta_{i^*,j^*} \nonumber \\
&\geq H_{\text{score}}(\mathbf d^{(l+1)}) - H_{\text{score}}(\mathbf d^{(l)}).
\end{align}

Because $\mathbf d^{(l+1)} \leq \mathbf d^{(l)}$, $G(\mathbf d^{(l+1)}$ also
does not contain a 2-cycle and, by Claim
\ref{claim:Hmax-monotonicity}, 
\begin{equation}
\min_{\mathbf y}H_{\text{max}}(\mathbf d^{(l+1)}, \mathbf y) \leq 
\min_{\mathbf y}H_{\text{max}}(\mathbf d^{(l)}, \mathbf y),
\end{equation}
which, together with the above, implies
\begin{align}
&{\min_{\mathbf y, \mathbf r}H(\mathbf d^{(l)}, \mathbf y, \mathbf r) - 
\min_{\mathbf y, \mathbf r}H(\mathbf d^{(l+1)}, \mathbf y, \mathbf r) } \nonumber\\
&= H_{\text{score}}(\mathbf d^{(l)}) - H_{\text{score}}(\mathbf d^{(l+1)})  \nonumber\\
&\quad + \min_{\mathbf y}H_{\text{max}}(\mathbf d^{(l)}, \mathbf y) - 
       \min_{\mathbf y}H_{\text{max}}(\mathbf d^{(l+1)}, \mathbf y)  \nonumber\\
&\quad + \min_{\mathbf r}H_{\text{cycle}}(\mathbf d^{(l)}, \mathbf r) - 
       \min_{\mathbf r}H_{\text{cycle}}(\mathbf d^{(l+1)}, \mathbf r)  \nonumber\\
&\geq H_{\text{score}}(\mathbf d^{(l)}) - H_{\text{score}}(\mathbf d^{(l+1)})  \nonumber\\
&\quad + H_{\text{trans}}(\mathbf r^{(l)}) - H_{\text{trans}}(\mathbf
r^{(l+1)}) \nonumber\\
&> 0.
\end{align}
Let $\mathbf d^{(0)} \equiv \mathbf d$. 
Because there can be only finitely many directed triangles in a graph of fixed order, we can construct a sequence 
$\mathbf d, \mathbf d^{(1)}, \ldots, \mathbf d^{(l)}, \ldots, \mathbf d^{(l^*)}$ 
such that 
$\min_{\mathbf y, \mathbf r}H(\mathbf d, \mathbf y, \mathbf r) > 
 \min_{\mathbf y, \mathbf r}H(\mathbf d^{(l^*)}, \mathbf y, \mathbf r)$ 
and $G(\mathbf r_{l^*})$ does not contain directed triangle. 
Thus 
$H_{\text{trans}}(\mathbf r_{l^*}) 
= H_{\text{cycle}}(\mathbf d^{(l^*)}, \mathbf r_{l^*}^*) = 0$, 
which means that $G(\mathbf d^{(l^*)})$ is a DAG. 
Setting $\mathbf d'\equiv \mathbf d^{(l^*)}$ completes the proof.

\end{proof}

\hcyclepenalty*
\begin{proof}
If $G(\mathbf d)$ contains a 2-cycle, then by Claim \ref{claim:2-cycle-removal}, there is some $\mathbf d''\leq \mathbf d$ such that $G(\mathbf d'')$ does not contain a directed 2-cycle and 
\begin{equation}
\min_{\mathbf y, \mathbf r}H(\mathbf d, \mathbf y, \mathbf r) > 
\min_{\mathbf y, \mathbf r}H(\mathbf d'', \mathbf y, \mathbf r).
\end{equation}
If $G(\mathbf d'')$ is a DAG, then setting $\mathbf d' \equiv \mathbf d''$
completes the proof. If $G(\mathbf d)$ does not contain a 2-cycle, set
$\mathbf d''\equiv \mathbf d$. Then by Claim
\ref{claim:transitivity-penalty-bound}, there is a $\mathbf d'\leq \mathbf d''$ such that $G(\mathbf d')$ is a DAG and 
\begin{equation}
\min_{\mathbf y, \mathbf r}H(\mathbf d'', \mathbf y, \mathbf r) > 
\min_{\mathbf y, \mathbf r}H(\mathbf d', \mathbf y, \mathbf r).
\end{equation}
\end{proof}

\overallpenalty*
\begin{proof}
Consider an arbitrary $\mathbf{d}$. 
If $\max_i|\mathbf d_i| > m$, then by Claim \ref{claim:max-penalty-bound}, there exists some $\mathbf{d}'$ such that
\begin{equation}
\min_{\mathbf y, \mathbf r}H(\mathbf d, \mathbf y, \mathbf r) > 
\min_{\mathbf y, \mathbf r}H(\mathbf d', \mathbf y, \mathbf r)
\end{equation}
and $\min_{\mathbf y}H_{\text{max}}(\mathbf d')=0$, i.e. $\max_i |\mathbf d'_i|\leq m$. If $G(\mathbf d)$ has a directed cycle, then by Claim \ref{claim:cycle-penalty-bound}, there is some $\mathbf d''$ such that 
\begin{equation}
\min_{\mathbf y, \mathbf r}H(\mathbf d, \mathbf y, \mathbf r) > 
\min_{\mathbf y, \mathbf r}H(\mathbf d'', \mathbf y, \mathbf r)
\end{equation}
and $\min_{\mathbf r}H_{\text{cycle}}(\mathbf d'', \mathbf r)=0$, i.e. $G(\mathbf d'')$ is a DAG. Either of these cases implies that $\mathbf d\neq \mathbf d^*$, so it must be that $G(\mathbf d^*)$ is a DAG, $\max_i|\mathbf d_i^*| \leq m$, and
$H(\mathbf{d}^*, \mathbf y, \mathbf r) = 
\min_{\substack{\max_i|\mathbf{d}_i|\leq m \\ G(\mathbf{d})\, \textnormal{is a DAG}}} H_{\textnormal{score}}(\mathbf{d}, \mathbf y, \mathbf r)$.
\end{proof}

\bibliographystyle{unsrt}
\bibliography{bn,qa,solar_flares,bio}
\end{multicols}

\end{document}